\documentclass[10pt,fleqn,leqno,dvipsnames]{article}
\usepackage[margin=1 in]{geometry}
\usepackage{amsmath,amssymb,amsthm}
\usepackage[colorlinks=true,linkcolor=black, citecolor=blue, urlcolor=blue, hypertexnames=false]{hyperref}
\usepackage{palatino,pxfonts,wasysym}
\usepackage{indentfirst,geometry,setspace,fancyhdr,sectsty}
\usepackage{relsize,url}
\usepackage{multirow}
\usepackage{dirtytalk}
\usepackage{multirow,array}

\usepackage{natbib}
\newcommand{\mybibliography}{\begin{singlespace}\renewcommand{\baselinestretch}{1}\small\normalsize\bibliography{mybibdata}\end{singlespace}}

\usepackage{graphicx}
\usepackage{adjustbox}
\usepackage[justification=centering,font=footnotesize]{caption}
\usepackage[labelformat=simple]{subfig}

\usepackage[capitalize,noabbrev,nameinlink]{cleveref}

\numberwithin{equation}{section}

\fancypagestyle{plain}{\cfoot{}}

\allsectionsfont{\mdseries}

\makeatletter
\def\@seccntformat#1{\@ifundefined{#1@cntformat}{\csname the#1\endcsname\hspace{.5em}}{\csname #1@cntformat\endcsname}}
\def\section@cntformat{\thesection.\hspace{.5em}}
\makeatother

\allowdisplaybreaks

\newcounter{mylistcounter}

\makeatletter
\newcommand\lilhead{\@startsection{paragraph}{4}{0em}{.1in}{.07in}{\bfseries\itshape}}
\makeatother

\newtheorem{theorem}{Theorem}

\newtheorem{corollary}{Corollary}

\newtheorem{claim}{Claim}

\makeatletter
\newtheorem*{rep@theorem}{\rep@title}
\newcommand{\newreptheorem}[2]{%
\newenvironment{rep#1}[1]{%
 \def\rep@title{#2 \ref{##1}}%
 \begin{rep@theorem}}%
 {\end{rep@theorem}}}
\makeatother

\newreptheorem{theorem}{Theorem}
\newreptheorem{proposition}{Proposition}
\newreptheorem{corollary}{Corollary}
\newreptheorem{lemma}{Lemma}
\newreptheorem{observation}{Observation}
\newtheoremstyle{examplestyle}{}{}{}{}{\itshape}{}{.5em}{\thmname{#1}\thmnumber{ #2.}\thmnote{ #3.}}
\theoremstyle{examplestyle}
\DeclareFontFamily{OT1}{pzc}{}
\DeclareFontShape{OT1}{pzc}{m}{it}{<-> s * [1.200] pzcmi7t}{}
\DeclareMathAlphabet{\mathscr}{OT1}{pzc}{m}{it}

\renewcommand{\implies}{\Rightarrow}

\renewcommand{\ast}{{\mathlarger *}}%makes the Palatino asterisk the correct size in formulas

\usepackage[textwidth=30mm]{todonotes}
\urldef{\hvte}\url{ratul.lahkar@ashoka.edu.in}
\urldef{\hvtw}\url{https://sites.google.com/view/ratul}
\urldef{\sae}\url{arigapudi@ssc.wisc.edu}
\urldef{\saw}\url{https://sites.google.com/view/srinivasarigapudi/home}

\title{Uniqueness of Inflection Points in Binomial\\ Exceedance Function Compositions}

\author{Srinivas Arigapudi\thanks{Department of Economic Sciences, IIT Kanpur, \protect\protect\href{mailto:arigapudi@iitk.ac.in}{arigapudi@iitk.ac.in}.}
\and Yuval Heller\thanks{Department of Economics, Bar-Ilan University, \protect\protect\href{mailto:yuval.heller@biu.ac.il}{yuval.heller@biu.ac.il}.} \and Amnon Schreiber\thanks{Department of Economics, Bar-Ilan University, \protect\protect\href{mailto:amnon.schreiber@biu.ac.il}{amnon.schreiber@biu.ac.il}.}%\thanks{Acknowledgent here.}
}
\date{\today}
\begin{document}
\begin{singlespace}
\maketitle
\begin{abstract}
We examine functions representing the cumulative probability of a binomial random variable exceeding a threshold, expressed in terms of the success probability per trial. These functions are known to exhibit a unique inflection point. We generalize this property to their compositions and highlight its applications.
\\

\end{abstract}
\end{singlespace}

\section{Introduction}
Let $k$ be a positive integer and $p\in [0,1].$ Let $X_k^p$ be a random variable following a binomial distribution with parameters $k$ and $p.$ Let $m\in \{1,2,\dots,k\}.$ Consider the function $F_{k,m}:[0,1]\rightarrow[0,1]$ representing the cumulative probability of the binomial random variable $X_k^p$ exceeding the threshold $m$:
\begin{equation}\label{eq:F_{k,m}_func}
    F_{k,m}(p) := Pr(X_k^p \geq m) \text{~~~~~~~~~~~~~~~~(Binomial exceedance function)}.
\end{equation}

Binomial exceedance functions naturally arise in various fields, such as hypothesis testing, risk assessments, epidemiology, and social learning. As shown in \cite{green1983fixed}, each binomial exceedance function has at most one interior
inflection point, which implies the existence of at most one interior fixed point. An interesting implication of this result is the following: if the probability that a biased coin yields at least $m$ heads in $k$ independent tosses equals the probability of yielding heads in a single toss, then this probability is uniquely determined.

Our main result (Theorem \ref{thm:two_times_composition}) establishes that the property of having a unique interior
inflection point extends to the composition of two binomial exceedance functions.  Specifically, $F_{k_1,m_1}(F_{k_2,m_2}(p))$ has at most one  inflection point where $k_1,k_2$ are positive integers and $1\leq m_i\leq k_i$ for $i=1,2.$ One implication of our result is  the following: if two biased coins have unknown probabilities $p_1,p_2\in(0,1)$, and for each coin $i\in\{1,2\}$, the probability that it yields at least $m_i$ heads in $k_i$ independent tosses equals the probability of yielding heads in a single toss of the other coin, then these probabilities $p_1,p_2$ are uniquely determined.

In a companion paper (\citealp{arigapudi2023heterogeneous}), 
we presented a variant of Theorem \ref{thm:two_times_composition} within a specific social learning setup, where agents sample the past behavior of a few others.\footnote{
Sampling dynamics were introduced in \cite{osborne1998games,sethi2000stability,sandholm2001almost}. Recent contributions to this literature include \cite{oyama2015sampling, mantilla2018efficiency,sandholm2019best,sandholm2020stability,arigapudi2020instability,Raj,izquierdo2022stability}.}
Here, we present our mathematical result as a stand-alone contribution, as we believe it could be valuable in contexts beyond the social learning application. This paper includes a new result---the second part of Corollary \ref{thm:cor_fixed_point}, which characterizes conditions for the existence of exactly one interior fixed point--- and offers a discussion and a new conjecture regarding the composition of multiple binomial exceedance functions.

\section{Preliminary Results}\label{sec:Result}
For completeness, we first present the result that a binomial exceedance function admits a unique inflection point (an adaptation of the results presented in \citealp{green1983fixed}).

\begin{claim}\label{thm:known_result}
    Let $k>0$ be an integer and $m$ be an integer such that $0\leq m \leq k.$ The function $F_{k,m}(p)$ has at most one interior inflection point i.e., there exists at most a single $p^{\ast} \in (0,1)$ at which $F_{k,m}''(p^{\ast}) = 0.$ Further, if $2\leq m \leq k-1$, then $F_{k,m}(p)$ has exactly one interior inflection point.
\end{claim}
\begin{proof}
    We have,
    \begin{align*}
        F_{k,m}(p) &= Pr(X_k^p \geq m)= \sum_{j= m}^k \binom{k}{j}p^j(1-p)^{k-j}\\
        \implies F_{k,m}'(p) &= \sum_{j\geq m}\left( j\binom{k}{j}p^{j-1}(1-p)^{k-j}- (k-j)\binom{k}{j}p^j(1-p)^{k-j-1}\right)\\
        &= k\sum_{j= m}^k\left( \binom{k-1}{j-1}p^{j-1}(1-p)^{k-j} - \binom{k-1}{j}p^j(1-p)^{k-j-1}\right)\\
        &= k\sum_{j=m}^k \left(Pr(X_{k-1}^p =j-1) - Pr(X_{k-1}^p =j)\right)\\
        &= k Pr(X_{k-1}^p=m-1)= k\binom{k-1}{m-1}p^{m-1}(1-p)^{k-m}\\
        \implies F_{k,m}''(p) &= k\binom{k-1}{m-1}p^{m-2}(1-p)^{k-m-1}((m-1)(1-p) - (k-m)p)\\
        &= k\binom{k-1}{m-1}p^{m-2}(1-p)^{k-m-1}(m-1 - (k-1)p)
    \end{align*}
    From the above computations, it follows that $F_{k,m}''(p) \geq 0$ if and only if $p\leq \frac{m-1}{k-1}$ and that $F_{k,m}''(p) = 0$ if and only if $p= \frac{m-1}{k-1}.$ The proof is complete.
\end{proof}

From the above computations, we have
\begin{align}
    F_{k,m}'(p) &= k\binom{k-1}{m-1}p^{m-1}(1-p)^{k-m} \label{eq:eqn1}\\
    F_{k,m}''(p) &= k\binom{k-1}{m-1}p^{m-2}(1-p)^{k-m-1}(m-1 - (k-1)p)\label{eq:eqn2}
\end{align}
Eqs. \eqref{eq:eqn1} and \eqref{eq:eqn2} will be later used in the proof of our main result, Theorem \ref{thm:two_times_composition}.

The following claim shows that the number of interior fixed points is determined by the number of inflection points. Formally,

\begin{claim}\label{cla:atmost_one_fixed_point}
    Let $G:[0,1]\rightarrow [0,1]$ be a twice continuously differentiable function with $G(0) = 0$ and $G(1) = 1.$ If $G$ has at most one interior inflection point, then $G$ has at most one interior fixed point.
\end{claim}

\begin{proof}
    Suppose, for contradiction,  that there are two interior fixed points $p^{\ast}, q^{\ast}\in (0,1).$ Without loss of generality, assume that $p^{\ast}<q^{\ast}.$ 
    Let $H:[0,1]\rightarrow [0,1]$ be defined as follows:
    $H(p) = G(p)-p$. The fact that $G$ is twice continuously differentiable implies that $H$ is also twice continuously differentiable.  We have,
    \[
    H(0) = H(p^{\ast}) = H(q^{\ast})= H(1) =0
    \]
    By Rolle's theorem, it follows that there exists $r^{\ast}\in (0, p^{\ast}),$ $s^{\ast}\in (p^{\ast}, q^{\ast}),$ and $t^{\ast}\in (q^{\ast}, 1)$ such that $H'(r^{\ast}) = H'(s^{\ast}) = H'(t^{\ast})=0.$ Another application of Rolle's theorem shows that there exists $u^{\ast}\in (r^{\ast}, s^{\ast})$ and $v^{\ast}\in (s^{\ast}, t^{\ast})$ such that $H''(u^{\ast}) = H''(t^{\ast}) =0.$
    Also,
    \begin{align*}
        H(p) = G(p) - p &\implies H'(p) = G'(p)-1    \implies H''(p) =G''(p).
    \end{align*}
    It therefore follows that $G''(u^{\ast}) = G''(t^{\ast}) =0.$ But this contradicts the fact that $G$ has at most one interior inflection point. Therefore, $G$ can have at most one interior fixed point.
\end{proof}
Our final claim shows that under mild conditions (specifically, a zero derivative at 0 and 1), exceedance functions possess at least one interior fixed point.
\begin{claim}\label{cla:atleast_one_fixed_point}
    Let $G:[0,1]\rightarrow [0,1]$ be a 
    continuously differentiable function such that $G(0) = 0, G(1)=1,$ and $G'(0) = G'(1) =0.$ Then, $G$ has at least one interior fixed point.
\end{claim}

\begin{proof}
    Let $H:[0,1] \rightarrow [0,1]$ be defined as $H(p) = G(p)-p.$ Suppose, for contradiction, that $G$ has no interior fixed point. This implies that for all $p\in (0,1),$ either $H(p)<0$ or $H(p) >0.$ 

    Suppose $H(p)<0$ i.e., $G(p)<p$ for all $p\in (0,1).$ $G$ being continuously differentiable implies that $H$ is continuously differentiable. We have,
    \begin{align*}
        H(p) = G(p)-p &\implies H'(p) = G'(p)-1 \implies H'(1) =G'(1) -1.
    \end{align*}
    We now compute as follows:
    \begin{align*}
        H'(1) &= G'(1) - 1= \lim_{h\rightarrow 0^{+}}\left(\frac{G(1)-G(1-h)}{h}\right) - 1\\
        &> \lim_{h\rightarrow 0^{+}}\left(\frac{1-(1-h))}{h}\right) - 1 \ \  \ \ (\text{since} \ G(1) =1 \ \text{and} \ G(1-h) <1-h)\\
        &= 0
    \end{align*}
    From the above, it follows that $H'(1) >0.$ Also, $G'(1) =0$ implies that $H'(1) = -1.$ Clearly, this is a contradiction. We thus cannot have $H(p) <0$ for all $p\in (0, 1).$  Similarly, we can show that $H(p)>0$ for all $p\in (0,1)$ is not possible. It therefore follows that $G$ has at least one interior fixed point. 
\end{proof}

Claims \ref{thm:known_result}--\ref{cla:atleast_one_fixed_point} jointly imply that $F_{k,m}$ admits at most one fixed point (and exactly one fixed point if $2\leq m\leq k-1$). This is illustrated in Figure \ref{fig:single_func}.
\begin{figure}[h]
\begin{centering}
\caption{The function $F_{6,m}(p)=PR(X_6^p\geq m)$ for various values of $m.$ The dots denote the inflection points.}\label{fig:single_func}
\includegraphics[scale=0.5]{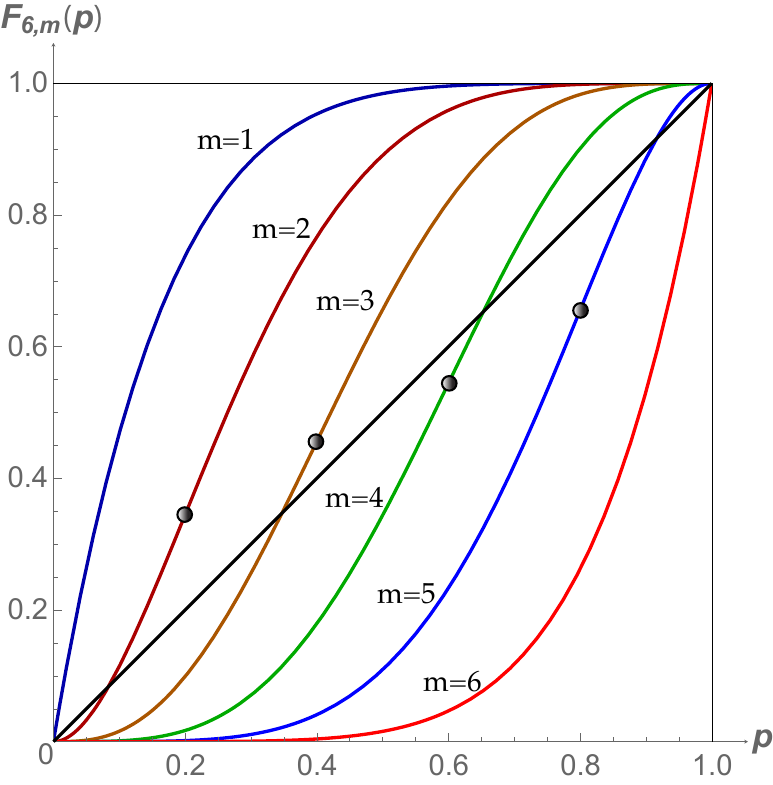}
\par\end{centering}
%{\small{}In the above plots, the $x$-axis denotes $p$ and the $y$-axis denotes the probability $F_{k,m}(p)$.}{\small\par}
\end{figure}

\section{Main Result}
Our main result shows that the property of having a unique inflection point extends to the composition of two binomial exceedance functions. 
\begin{theorem}\label{thm:two_times_composition}
    Let $k_1,k_2 > 0$ be integers. Let $m_1$ and $m_2$ be integers such that $1\leq m_1\leq k_1$ and $1\leq m_2\leq k_2.$ Let the function $F:[0,1]\rightarrow [0,1]$ be defined as follows:
    \begin{equation}\label{eq:composition_function}
        F(p) = F_{k_1,m_1}(F_{k_2,m_2}(p)).
    \end{equation}
    Then, $F(p)$ has at most one interior inflection point i.e., there exists at most a single $p^{\ast} \in (0,1)$ at which $F''(p^{\ast}) = 0.$     
\end{theorem}
\begin{proof}
    For $i\in \{1,2\},$ let $w_i:[0,1] \rightarrow [0,1]$ be defined as follows: $w_i(p) = F_{k_i, m_i}(p)$.
    
    For $i=1,2$ and $p\in (0,1),$ we have (from Eqs. \eqref{eq:eqn1} and \eqref{eq:eqn2})
    \begin{equation}\label{eq:w''/w'}
        \frac{w_i''(p)}{w_i'(p)} = \frac{k_i\binom{k_i-1}{m_i-1}p^{m_i-2}(1-p)^{k_i-m_i-1}(m_i-1 - (k_i-1)p)}{k_i\binom{k_i-1}{m_i-1}p^{m_i-1}(1-p)^{k_i-m_i}} = \frac{m_{i}-1}{p}-\frac{k_{i}-m_{i}}{1-p}.
    \end{equation}
    By definition, $F(p) = w_1(w_2(p)).$ For $p\in(0,1),$ using Eq. \eqref{eq:w''/w'}, we compute as follows:
    \begin{align*}
F'\left(p\right) &= w_{1}'\left(w_{2}\left(p\right)\right)w_{2}'\left(p\right)\\
F''\left(p\right)&=w_{1}''\left(w_{2}\left(p\right)\right)\left(w_{2}'\left(p\right)\right)^{2}+w_{1}'\left(w_{2}\left(p\right)\right)w_{2}''\left(p\right)\\
&= w_{1}'\left(w_{2}\left(p\right)\right)w_{2}'\left(p\right)\left[\left(\frac{m_{1}-1}{w_{2}\left(p\right)}-\frac{k_{1}-m_{1}}{1-w_{2}\left(p\right)}\right)w_{2}'\left(p\right)+\frac{m_{2}-1}{p}-\frac{k_{2}-m_{2}}{1-p}\right].
\end{align*}

The fact that each $w_{i}\left(p\right)$ is strictly increasing implies
that $F''\left(p\right)=0$ if and only if
\[
\left(\frac{m_{1}-1}{w_{2}\left(p\right)}-\frac{k_{1}-m_{1}}{1-w_{2}\left(p\right)}\right)w_{2}'\left(p\right)=\frac{k_{2}-m_{2}}{1-p}-\frac{m_{2}-1}{p}\Leftrightarrow
\]
\[
k_{2}\left(\begin{array}{c}
k_{2}-1\\
m_{2}-1
\end{array}\right)\left(\frac{m_{1}-1}{w_{2}\left(p\right)}-\frac{k_{1}-m_{1}}{1-w_{2}\left(p\right)}\right)p^{m_{2}-1}\left(1-p\right)^{k_{2}-m_{2}}=\frac{k_{2}-m_{2}}{1-p}-\frac{m_{2}-1}{p}\Leftrightarrow
\]
{\footnotesize{}
\[
\!\!\!\!\!\!\,\,\!\!k_{2}\left(\begin{array}{c}
k_{2}-1\\
m_{2}-1
\end{array}\right)\left(\frac{\left(m_{1}-1\right)p^{m_{2}}\left(1-p\right)^{k_{2}-m_{2}}}{\sum_{l=m_{2}}^{k_{2}}\left(\begin{array}{c}
k_{2}\\
l
\end{array}\right)p^{l}\left(1-p\right)^{k_{2}-l}}-\frac{\left(k_{1}-m_{1}\right)p^{m_{2}}\left(1-p\right)^{k_{2}-m_{2}}}{\sum_{l=0}^{m_{2}-1}\left(\begin{array}{c}
k_{2}\\
l
\end{array}\right)p^{l}\left(1-p\right)^{k_{2}-l}}\right)\frac{1}{p}=\frac{k_{2}-m_{2}}{1-p}-\frac{m_{2}-1}{p}\Leftrightarrow
\]
}
\vspace{-15px}
\begin{equation}\label{eq:bigeqn}
k_{2}\left(\begin{array}{c}
k_{2}-1\\
m_{2}-1
\end{array}\right)\left(\frac{m_{1}-1}{\sum_{l=m_{2}}^{k_{2}}\left(\begin{array}{c}
k_{2}\\
l
\end{array}\right)\left(\frac{p}{1-p}\right)^{l-m_{2}}}-\frac{k_{1}-m_{1}}{\sum_{l=0}^{m_{2}-1}\left(\begin{array}{c}
k_{2}\\
l
\end{array}\right)\left(\frac{1-p}{p}\right)^{m_{2}-l}}\right)\frac{1}{p}=\frac{k_2-m_2}{1-p}-\frac{m_2-1}{p}.
\end{equation}
Using the fact that the functions $\frac{p}{1-p}, \frac{1}{1-p}$ are strictly increasing in $p$ and that the functions $\frac{1-p}{p}, \frac{1}{p}$ are strictly decreasing in $p,$ one can easily verify that the left-hand side of Eq. \eqref{eq:bigeqn} is strictly decreasing and that the right-hand side is strictly increasing in $p.$ Therefore, there can be at most one point $p^{\ast}\in (0,1)$ at which $F''(p^{\ast})=0$ i.e., $F$ has at most one interior inflection point.
\end{proof}

From Theorem \ref{thm:two_times_composition}, Claims \ref{cla:atmost_one_fixed_point} and \ref{cla:atleast_one_fixed_point}, it follows that the composition of two binomial exceedance functions admits at most a single interior fixed point. This is illustrated in  Figure \ref{fig:composition_func}.
\begin{figure}[h]
\begin{centering}
\caption{The function $F_{3,m_1}(F_{4,m_2}(p))$ for various values of $m_1$ and $m_2.$ The dots denote the inflection points. }\label{fig:composition_func}
\includegraphics[scale=0.5]{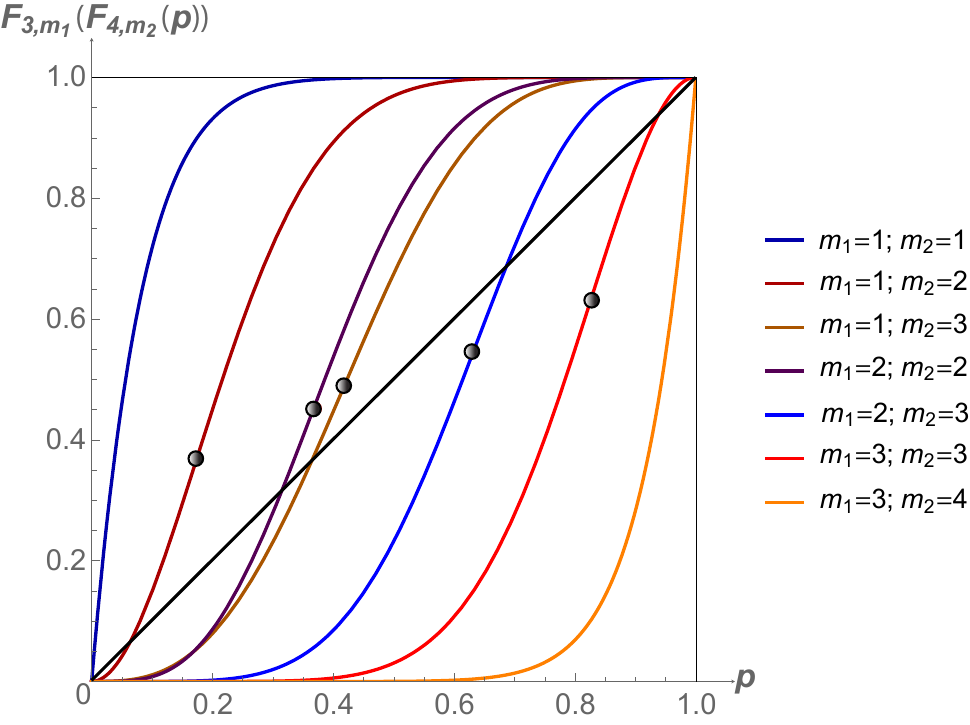}
\par\end{centering}
%{\small{}In the above plots, the $x$-axis denotes $p$ and the $y$-axis denotes the probability $F_{k,m}(p)$.}{\small\par}
\end{figure}
\begin{corollary}\label{thm:cor_fixed_point}
    Let $k_1,k_2 > 0$ be integers. Let $m_1$ and $m_2$ be integers such that $0<m_1\leq k_1$ and $0<m_2\leq k_2.$ 
    The function $F(p)$ defined in Eq. \eqref{eq:composition_function} has at most one interior fixed point i.e., at most a single $p_{fix} \in (0,1)$ at which $F(p_{fix}) = p_{fix}.$          
        Further, if either $2\leq m_1 \leq k_1-1$ or $2\leq m_2 \leq k_2-1$ then $F(p)$ has exactly one interior fixed point.
\end{corollary}
\begin{proof}
    Recall, $F(p) = w_1(w_2(p)),$ where $w_i(p) = Pr(X_{k_i}^p \geq m_i)$ for $i=1,2.$ Clearly, $w_i(0) =0$ and $w_i(1) =1$ for $i=1,2.$ We have,
    \begin{equation}\label{eq:Fprime_0_1}
      F'(p) = w_1'(w_2(p))w_2'(p) \ \implies F'(0) = w_1'(0) w_2'(0), \ \ F'(1) = w_1'(1) w_2'(1).  
    \end{equation}
    If $2\leq m_1\leq k_1-1,$ then $w_1'(0) = w_1'(1) =0.$ Similarly, if $2\leq m_2\leq k_2-1,$ then $w_2'(0) = w_2'(1) =0.$ In either case, from Eq. \eqref{eq:Fprime_0_1}, it follows that $F'(0) = F'(1) = 0.$ The statement now follows from Theorem \ref{thm:two_times_composition}, Claims \ref{cla:atmost_one_fixed_point} and \ref{cla:atleast_one_fixed_point}.
\end{proof}
\section{Discussion}
\subsection{Uniqueness of Inflection Points is not Invariant to Compositions}
The following example illustrates that the composition of general functions with unique inflection points might admit multiple inflection points. This highlights that the property that the uniqueness of inflection points is preserved under composition is a special property of binomial exceedance functions.

Define functions $g_1,g_2:[0,1]\rightarrow [0,1]$ as follows:
\begin{align*}
    g_1(p) &= 4\left(p-\frac{1}{2}\right)^3 + \frac{1}{2}, \ \ g_2(p) = \frac{64}{28}\left(p-\frac{3}{4}\right)^3 + \frac{27}{28}.
\end{align*}
It is easily verified that $g_1$ has a unique inflection point at $p=\frac{1}{2}$ and $g_2$ has a unique inflection point at $p=\frac{3}{4}.$ Also, $g_i(0) =1, g_i(1) = 0$ for $i=1,2.$  The composition function $g_1(g_2(p)),$ however has more than one interior inflection point as can be seen in Figure \ref{fig:counter_example}.

\begin{figure}[h]
\begin{centering}
\caption{Example of Composition with Multiple Inflection Points}\label{fig:counter_example}
\includegraphics[scale=0.5]{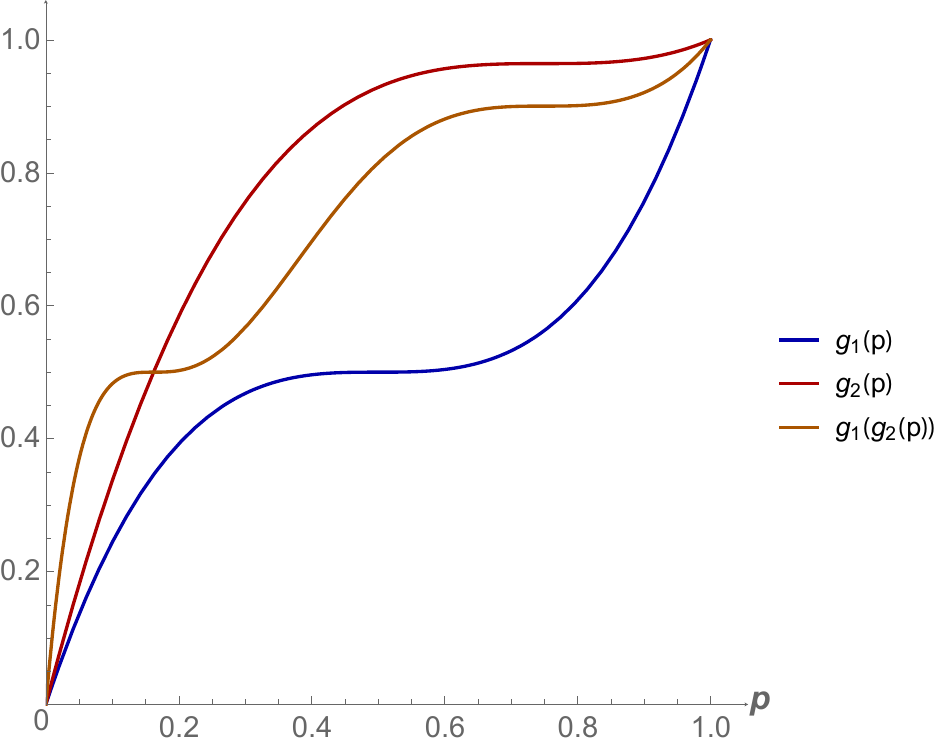}
\par\end{centering}
\end{figure}

\subsection{Conjecture}
We conjecture that the uniqueness of the inflection point of binomial exceedance functions is preserved also under composition of $n>2$ functions. 

Formally, let $k_1,k_2,\dots,k_n$ be positive integers and $m_1,m_2,\dots,m_n$ be integers such that $1\leq m_i\leq k_i.$ For $i=1,2,\dots,n,$ define $w_i:[0,1]\rightarrow [0,1]$ as follows:
\[
w_i(p) = Pr(X_{k_i}^p \geq m_i).
\]
Define the function $F:[0,1]\rightarrow [0,1]$ as follows:
\begin{equation}\label{eq:conjecture_func}
    F(p) = w_1(w_2(w_3(\dots w_n(p)\dots)))
\end{equation}
We conjecture that the function in Eq. \eqref{eq:conjecture_func} has at most one interior inflection point.

\mybibliography

\begin{thebibliography}{}

\bibitem[Arigapudi et~al., 2021]{arigapudi2020instability}
Arigapudi, S., Heller, Y., and Milchtaich, I. (2021).
\newblock Instability of defection in the prisoner’s dilemma under best experienced payoff dynamics.
\newblock {\em Journal of Economic Theory}, 197:105174.

\bibitem[Arigapudi et~al., 2024]{arigapudi2023heterogeneous}
Arigapudi, S., Heller, Y., and Schreiber, A. (2024).
\newblock Heterogeneous noise and stable miscoordination.
\newblock {\em American Economic Journal: Microeconomics (forthcoming)}.

\bibitem[Green, 1983]{green1983fixed}
Green, D.~N. (1983).
\newblock Fixed point properties of the binomial function.
\newblock {\em Journal of the Franklin Institute}, 316(1):51--62.

\bibitem[Izquierdo and Izquierdo, 2022]{izquierdo2022stability}
Izquierdo, S.~S. and Izquierdo, L.~R. (2022).
\newblock Stability of strict equilibria in best experienced payoff dynamics: Simple formulas and applications.
\newblock {\em Journal of Economic Theory}, 206:105553.

\bibitem[Mantilla et~al., 2018]{mantilla2018efficiency}
Mantilla, C., Sethi, R., and C{\'a}rdenas, J.~C. (2018).
\newblock Efficiency and stability of sampling equilibrium in public goods games.
\newblock {\em Journal of Public Economic Theory}, 22(2):355--370.

\bibitem[Osborne and Rubinstein, 1998]{osborne1998games}
Osborne, M.~J. and Rubinstein, A. (1998).
\newblock Games with procedurally rational players.
\newblock {\em American Economic Review}, 88(4):834--847.

\bibitem[Oyama et~al., 2015]{oyama2015sampling}
Oyama, D., Sandholm, W.~H., and Tercieux, O. (2015).
\newblock Sampling best response dynamics and deterministic equilibrium selection.
\newblock {\em Theoretical Economics}, 10(1):243--281.

\bibitem[Sandholm, 2001]{sandholm2001almost}
Sandholm, W.~H. (2001).
\newblock Almost global convergence to $p$-dominant equilibrium.
\newblock {\em International Journal of Game Theory}, 30(1):107--116.

\bibitem[Sandholm et~al., 2019]{sandholm2019best}
Sandholm, W.~H., Izquierdo, S.~S., and Izquierdo, L.~R. (2019).
\newblock Best experienced payoff dynamics and cooperation in the centipede game.
\newblock {\em Theoretical Economics}, 14(4):1347--1385.

\bibitem[Sandholm et~al., 2020]{sandholm2020stability}
Sandholm, W.~H., Izquierdo, S.~S., and Izquierdo, L.~R. (2020).
\newblock Stability for best experienced payoff dynamics.
\newblock {\em Journal of Economic Theory}, 185:104957.

\bibitem[Sethi, 2000]{sethi2000stability}
Sethi, R. (2000).
\newblock Stability of equilibria in games with procedurally rational players.
\newblock {\em Games and Economic Behavior}, 32(1):85--104.

\bibitem[Sethi, 2021]{Raj}
Sethi, R. (2021).
\newblock Stable sampling in repeated games.
\newblock {\em Journal of Economic Theory}, 197:105343.

\end{thebibliography}
\end{document}